\documentclass[11pt,letterpaper]{article}
\usepackage[a4paper,margin=2cm]{geometry}

\usepackage{graphicx}
\usepackage[utf8]{inputenc}
\usepackage{color}
\usepackage{float}
\usepackage{amsmath}
\usepackage{amssymb}
\usepackage{esint}
\usepackage{hyperref}
\usepackage{varioref}

 \usepackage{amsthm}

\newtheorem{theorem}{Theorem}[section]
\newtheorem{corollary}{Corollary}[section]
\newtheorem{lemma}{Lemma}[section]
\newtheorem{definition}{Definition}

\newtheorem{remark}{Remark}[section]

\makeatletter

\usepackage{caption}

\usepackage{array}
\usepackage{IEEEtrantools}
\usepackage{pdfpages}
\usepackage{tikz}
\usepackage{algorithm}
\usepackage[noend]{algpseudocode}
\usepackage{algorithmicx}
\let\OldStatex\Statex
\renewcommand{\Statex}[1][3]{%
  \setlength\@tempdima{\algorithmicindent}%
  \OldStatex\hskip\dimexpr#1\@tempdima\relax}

\def\ff{\mathbb{F}}

\def\fq{\mathbb{F}_q}
\def\fqn{\mathbb{F}_q^n}
\def\fqN{\mathbb{F}_q^N}

\def\Vo{V^o}

\def\mcb{\mathcal{B}}
\def\beq{\begin{equation}}
\def\eeq{\end{equation}}
\def\yop{y_{op}}
\def\phat{\hat{\psi}}
\def\hy{\hat{y}}

\begin{document}

\title{Koopman operator approach for computing structure of solutions and Observability of non-linear finite state systems
}


\author{Ramachandran Anantharaman\\ramachandran@ee.iitb.ac.in \and Virendra Sule\\vrs@ee.iitb.ac.in \and
Department of Electrical Engineering,\\ 
Indian Institute of Technology Bombay, India.}

\maketitle

\begin{abstract}

Given a discrete dynamical system defined by a map in a vector space over a finite field called Finite State Systems (FSS), a dual linear system over the space of functions on the state space is constructed using the dual map. This system constitutes the well known Koopman linear system framework of dynamical systems, hence called the Koopman linear system (KLS). It is first shown that several structural properties of solutions of the FSS can be inferred from the solutions of the KLS. The problems of computation of structural parameters of solutions of non-linear FSS are computationally hard and hence become infeasible as the number of variables increases. In contrast, it has been well known that these problems can be solved by linear algebra for linear FSS in terms of elementary divisors of matrices and their orders. In the next step, the KLS is reduced to the smallest order (called RO-KLS) while still retaining all the information of the parameters of structure of solutions of the FSS. Hence when the order of the RO-KLS is sufficiently small, the above computational problems of non-linear FSS are practically feasible. Next, it is shown that the observability of the non-linear FSS with an output function is equivalent to that of the RO-KLS with an appropriate linear output map. Hence, the problem of non-linear observability is solved by an observer design for the equivalent RO-KLS. Such a construction should have striking applications to realistic FSS arising in Cryptology and Biological networks.    
\end{abstract}

\section*{Notations}
A finite field with $q$ elements is denoted as $\fq$ where $q = p^d$ for a prime $p$, called characteristic of the field and $d$, the degree of extension of the prime field $\ff_p$. $\fqn$ denotes the vector space of $n$-tuples over $\fq$ and $\Vo$ is the vector space of $\fq$-valued functions over $\fqn$. 

\section{Introduction}
\label{sec:intro}
Finite State Systems (FSS) are discrete time dynamical systems which evolve over a finite set $X$. Such dynamical systems arise in vast number of applications such as Stream Cipher algorithms which are used in generation of pseudo random sequences and encryption \cite{Golomb,Goresky}, in Systems Biology to model dynamics of Genetic Regulatory Networks, modeling of Biochemical reactions \cite{Kauffman,Goodwin,Kauffman2}, in Computer Science in modeling finite state machines \cite{Gill2} etc. In this paper we develop an application of an approach using the Koopman operator of these systems to solve two problems. The first problem is concerned with computing the structure of solutions of FSS while the second problem is to solve the \emph{observability} problem of non-linear FSS. These systems are considered evolving over the state space $\ff_q^n$ with a non-linear output map taking values in $\ff_q^m$. Such a dynamical system is specified by equations
\begin{equation}
\label{eq:FSS}
  \begin{aligned}
    x(k+1) &= F(x(k)) \\
    z(k) &= g(x(k)) 
    \end{aligned}
\end{equation}
where $F$ and $g$ known as state update and output maps respectively and are polynomial maps\footnote{It is well known that any map from $\fqn \to \fqn$ for any finite field $\ff_q$ is a polynomial map \cite{Lidl}} over $\fqn$. Such a system is said to be \emph{observable} if given a sequence of its outputs $z(k_0),z(k_0+1),\dots,z(k_0+L)$ for some $k_0,L$, there exists an unique initial condition $x(k_0)$ which generates the output sequence. 
\subsection{Koopman operator approach to Finite State Systems}
Koopman operator associated with the state update map $F$ is the linear operator  ($\Phi$) in the space of functions $\Vo$ over the state space $\ff_q^n$ defined as 
\[
\Phi \psi(x) = \psi(F(x))
\]
for $\psi \in \Vo$. However due to the exponential size of $\Vo$, such a map is hardly useful for computation in realistic situations unless $n$ is small. In this paper we show that a reduction of the linear system defined by the Koopman operator exists which can make this computation feasible. This resolves the problem of solving several hard problems of computation of structure of solutions of FSS whenever this reduction happens to have small enough dimension. Although the Koopman operator is known for analysis of dynamical systems in continuous time, such conditions for feasibility of computations of these hard problems in the case of FSS do not seem to have been explored earlier.

\subsection{Koopman operator approach to observability}
We show in this paper that the above reduction of the Koopman operator also resolves the problem of observability and design of an observer for non-linear FSS with output maps by practically feasible computation whenever the dimension of the reduced operator is small enough. The resulting construction of the observer for non-linear FSS using linear observer theory as reported in this paper is also believed to be unexplored in the previous literature.

\subsection{Structure of solutions of FSS}
Due to finiteness of the state space it easily follows that the solutions (or trajectories) of (\ref{eq:FSS}) are any one of the following three types only
\begin{enumerate}
    \item \emph{Fixed points}: points $x$ such that $F(x)=x$.
    \item \emph{Closed orbits}: sequences $x(k)$ (also called just orbits), $k=0,1,2,\ldots$ given $x(0)$ such that $x(k+N)=x(k)\;\forall k$. The smallest such $N>0$ is called the \emph{orbit length} (or period) of the orbit through $x(0)$. 
    \item \emph{Chains}: sequences $x(k)$ for which there is $M$ such that $x(k+N)=x(k),\ \ N \geq 0, \ \ k > M$. The chains thus settle down into an orbit or a fixed point. The minimal $M$ is called the \emph{chain length} of the chain starting at $x(0)$. The point $x(0)$ is called the \emph{root} of the chain if $x(0) = F(x)$ has no solutions $x$ in the state space.  
\end{enumerate} 
Hence, given a FSS by its map $F$ in (\ref{eq:FSS}), nature of a solution passing through any point $x$ in the state space is determined by whether the point $x$ is on a closed orbit of length $\geq 1$ or on the attractor set of such closed orbits. Problems such as computing all fixed points of the map $F$, lengths of its periodic orbits, lengths of chains and their roots are computationally hard. For instance to determine orbits of length $l$ it is necessary to find solutions of the compositional equation $F^{(l)}(x)=x$. In the specific case of the base field $\ff_2$ and $F$ a non-linear map this is a problem all finding satisfying assignments of Boolean satisfiabilty problem. These problems are in general known to be \emph{NP-Hard}. Hence understanding cases in which computation of solutions to these problems is practically feasible is an important issue in research on FSS. It has been well known since the work of \cite{Gill1} that the linear FSS (when the map $F$ is $\ff_q$-linear) all the above problems are solvable in polynomial time in $n$. Hence the linear case of FSS is considered to be practically feasible for computations. Such a linear FSS can be represented by the equation  
\begin{equation}
    \label{eq:LFSS}
    x(k+1) = Ax(k)
\end{equation}
where $x(k)$ belong to $\ff_q^n$ and $A$ is an $n\times n$ matrix over $\ff_q$. As shown in \cite{Gill1}, the problems of computation of solutions of the LFSS can be completely solved in terms of elementary divisors of the matrix $A$ and their orders. The questions we address in this paper are, how to extract information about solutions of a non-linear FSS by constructing and analyzing a linear FSS? When is such a strategy computationally feasible without resorting to an exponential in $n$ size linear systems?

\subsection{Previous work on use of linear systems to solve non-linear problems}
\cite{Dcheng} proposed an approach for Boolean networks based on semi-tensor products to convert a Boolean dynamical system to a linear map over binary field $\ff_2$. This method gives rise to a matrix of exponential size $2^n$ for an $n$-bit Boolean network which is then analyzed to study the properties of the original Boolean dynamical system. By construction, this method leads to computations which are always of exponential complexity. Also the semi-tensor product based approach is very specific to address Boolean dynamical systems while the approach proposed in this paper can be used to address non-linear dynamical systems over general finite fields. 

In this paper we propose an alternative way to compute information on solutions of non-linear FSS using the dual linear FSS determined by the Koopman operator of the map $F$ which need not always result into an exponential size linear system. For non-linear differentiable continuous time dynamical systems the Koopman operator based analysis has been known since long \cite{Koopman,Surana,JHTu,Arbabi}. A finite dimensional reduced linear representation of Koopman operator was proposed in \cite{Brunton}. The method used for constructing the reduced order Koopman operator proposed in this paper is along similar directions. However the Koopman framework does not appear to have been explored for FSS over finite fields. 
 
 \subsection{Contributions of this paper}
 In this paper we develop an approach using the Koopman operator to formulate the following problems of FSS which can be potentially useful for solving these problems in practical cases.
 \begin{enumerate}
 \item First we formulate the construction of a dual linear system evolving under the Koopman operator (called KLS) on the space of functions on the state space for a given nonlinear map $F$ of an FSS. It is then shown that the KLS contains information of all trajectories of the FSS. Moreover state trajectories of the FSS are obtained as an image of the trajectories of the KLS.
 \item Due to unsatisfactory dimension of the KLS (barring exceptions) a reduction of the KLS is obtained by defining a linear system which evolves under the restriction of the Koopman operator on a cyclic invariant subspace which contains all co-ordinate functions. This linear system has state space which is the space of values of a basis of this cyclic subspace. This is called the reduced order KLS (RO-KLS). It is then shown that the information about trajectories of the FSS such as lengths of chains, periods of closed orbits and fixed points can be obtained from the elementary divisors of the linear map of the RO-KLS. Hence whenever the dimension of this system is of practically feasible size such as $O(n^k)$ for $k\leq 3$ where $n$ is the number of state variable of the FSS, the computations of parameters of trajectories and fixed points of the FSS is practically feasible by linear algebra. This is a nontrivial observation and has no analogous correspondence in the Koopman operator theory of continuous non-linear systems.
 \item We then make an important application of the above theory to formulate and solve the problem of observer design and observability for non-linear FSS using the linear system RO-KLS and the output map of FSS. Hence this resolves the problem of observer design for non-linear FSS. Hence the construction of an observer is feasible for FSS whenever the RO-KLS has feasible size. This is an important systems theoretic result and to the best of the knowledge of authors, no such construction has been proposed for non-linear FSS in the previous literature. The observer construction has a potential for application to several applied problems in Cryptology, Biological and Biochemical networks. These applications shall be pursued elsewhere.    
 \end{enumerate}

The theory of FSS developed in this paper is applicable over general finite fields hence can also be considered for application to Boolean networks without explicit use of properties of Boolean functions.


\section{Dual of FSS: Koopman Linear System}
\label{sec:DualFSS}
The dual system of the FSS (\ref{eq:FSS}) is a linear FSS whose state space is the vector space of functions on the state space. Let $V^0$ denote the vector space of $\ff_q$ valued functions on $\fqn$. The elements of $V^0$ are called \emph{observables} of the FSS. The Koopman operator is the map $\Phi:V^0\rightarrow V^0$ defined by
\beq 
\label{eq:DualSys}
\Phi \psi(x)=\psi\circ F(x)\ \ \mbox{ for}\, \psi\in V^0
\eeq
which defines a linear dynamical system 
\beq
\label{eq:Koopman}
\psi_{k+1}(x) = \Phi \psi_k(x) = \psi_k(F(x)) \hspace{0.3in} 
\eeq
We call this FSS as the Koopman Linear System (KLS) associated with the FSS (\ref{eq:FSS}). Given an observable $\psi \in V^0$ the evolution of $\psi$ under the Koopman operator $\Phi$ generates the sequence 
\[
\psi(x),\psi(F(x)),\psi(F^2(x)),\dots, \psi(F^l(x)),\dots
\]
in $\Vo$ and given an initial point $x_0$ in $\fqn$, evaluation of functions in this sequence gives the sequence
\[
\psi(x_0),\psi(F(x_0)),\psi(F^2(x_0)),\dots, \psi(F^l(x_0)),\dots
\]
in $\ff_q$ called the sequence of \emph{evaluations} of the sequence of observables.

\subsection{Relationship of orbits of KLS with that of FSS}
As the KLS is also an FSS over the state space $\Vo$, the fixed points, orbits and chains of KLS can be defined similarly as any FSS except that fixed points are observables and orbits and chains of KLS are orbit and chains of observables. We show how trajectories of KLS reveal properties of trajectories of FSS. 
\newline

\begin{lemma} 
If the FSS has an orbit of length $l$, then the KLS also has an orbit of length $l$.
\label{lem:OrbitsKLS}
\end{lemma}
\begin{proof}
 Recall that KLS is defined by action of $\Phi$ on space of functions by $\Phi \psi(x) = \psi(F(x))$. Let $x_0$ be a point in $\fqn$ which lies on a closed orbit of length $l$ under the action of $F$. Let $S_l$ be the set defined as follows:
\[
S_l = \{ x \in \fqn \ \ |\  F^{kl}(x) = x_0, k \in \mathbb{Z}_+\}
\]
This set includes $x_0$ and all other points on the chains terminating in the orbit containing $x_0$ and pass through $x_0$ under the action of $F^{kl}$. 

Note that the set $S_l$ is invariant under action of $F^l$. Since for any $x \in S_l$, $\exists$ $k_0$ such that $F^{k_0l}(x) = x_0$. Hence $F^{k_0l}(F^l(x)) = F^{(k_0+1)l}(x) \in S_l$. 

Similarly if $x \in S_l^c$, the complement of $S_l$, then $F^{kl}(x) \neq x_0$ for any $k$. Hence $F^{kl}$ acting on $F^l(x)$ gives 
\begin{align*}
    F^{kl}(F^l(x)) &= F^{(k+1)l} (x) \neq x_0
\end{align*}
Hence $F^l(x) \notin S_l$ whenever $x \notin S_l$ which shows $S_l^c$ is also $F^l$ invariant. 


Construct a function $\psi \in \Vo$ defined as follows. 
\[
\psi (x) = \left\{ \begin{matrix} 1 & \forall\ x \in S_l \\ 0 & \forall \ x \in S_l^c\end{matrix} \right.
\]
We claim that $\psi$ is one such function which has a closed orbit of length $l$ under action of $\Phi$. To prove this, we first prove that $\psi$ lies on a periodic orbit and secondly we prove that $l$ is the period of $\psi$.  Now,
\begin{align*}
    \Phi^{l}\psi(x) &= \psi(F^l(x))  \\
    &= \left\{ \begin{matrix} 1 & \mbox{for} & x \in S_l \\ 0 & \mbox{for} & x \in S_l^c\end{matrix} \right. \\
    &= \psi(x)
\end{align*}
The second equality, comes due to the fact that $S_l$ and $S_l^c$ are invariant sets under $F^l$. This proves that $\psi(x)$ has a periodic orbit under $\Phi$ and whose orbit length divides $l$.

Let $0 < m < l$ be the orbit length of $\psi$. We prove that this leads to a contradiction.  Since $\Phi^m \psi(x) = \psi(x)$, in particular when $x = x_0$,
\begin{align}
    \Phi^m \psi(x_0) &= \psi(x_0) \nonumber \\
    \implies \psi(F^m(x_0)) &= \psi(x_0) \label{eq:dummy1}
\end{align}
But by definition $\psi(x) = 1$ only if $x \in S_l$. So, in (\ref{eq:dummy1}), RHS equals 1. To get the contradiction, we prove that $F^m(x_0) \notin S_l$. 

Suppose let $\alpha := F^m(x_0) \in S_l$, then there exists some $k_0$ such that
\[
F^{k_0l}(\alpha) = x_0
\]
but $\alpha = F^m(x_0)$. Substituting back,
\begin{align*}
\begin{array}{rrl}
    &F^{k_0l}(F^m (x_0)) &= x_0 \\
    \implies& F^{k_0l+m}(x_0) &= x_0 \\
    \implies&F^{m}(F^{k_0l}(x_0)) &= x_0 \\
    \implies& F^m(x_0) &= x_0
\end{array}
\end{align*} 
The last equation is due to the fact that $x_0$ is on an closed orbit of length $l$ and leads to a contradiction since $l$ is the least integer such that $f^l(x_0) = x_0$ and $m < l$ by assumptions. Hence $F^m(x_0) \notin S_l$ and in (\ref{eq:dummy1}), LHS = 0. This is a contradiction. Hence $\psi$ can not have an orbit of length $m < l$. Hence the orbit length of $\psi$ under $\Phi$ is precisely $l$.

This constructs a specific $\psi \in \Vo$ which has an orbit length $l$ under KLS whenever the FSS has an orbit of length $l$.  
\end{proof}
\begin{definition}
A FSS is called \emph{non-singular} if the map $F$ is bijective in $\fqn$. 
\end{definition}
All trajectories of non-singular FSS are either fixed points or periodic orbits. 
\newline 
\begin{lemma}
\label{lem:KLS-FixedPt}
If the FSS is non-singular, then the KLS is also non-singular. 
\end{lemma}
This means that when the FSS is non-singular, the linear map $\Phi$ is bijective over $\Vo$.

\begin{proof}
Given that the FSS is non-singular, we have for $x_1,\ x_2 \in V$, 
\[
F(x_1) \neq F(x_2) \ \ \ \forall \ x_1 \neq x_2
\]
Let $\psi_1,\ \psi_2 \in \Vo$ and $\psi_1 \neq \psi_2$. Let $\psi_d = \psi_1 - \psi_2$. We will prove that ``given $F$ is non-singular and if $\Phi\psi_1(x) = \Phi\psi_2(x) \ \forall \ x \in V$, then $\psi_d(x) = 0 \ \forall \ x \in V$" which proves that if FSS is non-singular, then Koopman operator is non-singular. 
\begin{align*}
\begin{array}{rrcl} & \Phi\psi_1(x)  &= &\Phi \psi_2(x) \ \ \forall\ x \in V \\
\implies& \Phi(\psi_1 - \psi_2)(x) &=& 0 \ \ \forall\ x \in V \\ 
\implies& \Phi \psi_d(x) &=& 0 \ \ \forall\ x \in V \\
\implies& \psi_d(F(x)) &=& 0 \ \ \forall\ x \in V  \end{array}
\end{align*}
This implies that $\psi_d(F(x)) = 0$ for all $x \in V$. But since $F$ is non-singular, the image of $F$ is the whole of $V$. So we have 
\[
\psi_d(y) = 0 \ \ \forall \ y \in V
\]
This means that $\psi_d$ is the \emph{zero} function which leads to a contradiction because $\psi_d = \psi_1 - \psi_2$ which is non-zero. 
\end{proof}

The following lemma gives the relationship between the chain (and their lengths) of the FSS with the chains of the KLS. 
\newline
\begin{lemma}
\label{lem:ChainsKLS}
If the FSS has a chain of length $l$, then the KLS also has a chain of length $l$.
\end{lemma}
The proof is given in appendix. The above three lemmas together can be written as a theorem characterizing all the solutions of the FSS through the KLS. 
\newline 
\begin{theorem}
Given a FSS and the corresponding KLS
\begin{enumerate}
    \item If there exists an orbit of length $N$ in the FSS, then there exists an orbit of length $N$ in the KLS.
    \item If there exists a chain of length $l$ in the FSS, then there exists a chain of length $l$ in the KLS.
\end{enumerate}
\end{theorem}
The above theorem establishes that by constructing the KLS from an FSS, one does not miss any structure of solutions of FSS (i.e. lengths of orbits and chains of  solution of FSS are present in the set of lengths of solutions of the KLS). However KLS may have several extra solutions which do not correspond to any solutions of the FSS. This is because of the fact the KLS evolves over a bigger state space than the FSS. But the extra solutions of the KLS are not completely unrelated to the FSS. The following theory characterizes properties of all solutions of the KLS and its correspondence with the properties of solutions of FSS.
\newline 
\begin{definition}
Given a non-singular FSS, a number $\nu$ is called as the \emph{period of the FSS} if $\nu$ is the smallest positive integer such that 
\[
F^{\nu}(x) = x \ \ \ \forall\ \ x \in \ \fqn
\]
\end{definition}
Clearly, due to finiteness of the number of orbits, the period of a non-singular FSS is the least common multiple of all the possible orbit lengths of the FSS. Next result relates period of FSS with that of the KLS
\newline 
\begin{lemma} \label{lem:FtoKOrbits}
Given a non-singular FSS, the period of FSS is equal to the period of the associated KLS.
\end{lemma}
\begin{proof}
 Let the FSS have a period $n_1$ and the KLS have a period $n_2$. We prove that $n_1|n_2$\footnote{$n_1|n_2$ means $n_2$ is an integer multiple of $n_1$} and $n_2|n_1$ and thereby prove that $n_1 = n_2$. Since the period of FSS is $n_1$,
 \[
 x(k+n_1) = F^{n_1}x(k) = x(k) \ \ \forall \ \ x(k) \ \in \ \fqn
 \]
 For some $\psi_0 \in \Vo$, 
 \begin{align}
 &\Phi^{n_1} \psi_0(x) = \psi_0 (F^{n_1}(x)) = \psi_0(x) \nonumber  \\
 &\implies n_1 = ln_2 \ \mbox{for some} \ l\ \in\ \mathbb{Z}_+ \label{eq:temp1}
 \end{align}
 Similarly, since the period of KLS is $n_2$,
\[
\Phi^{n_2} \psi_0(x) = \psi_0(x)\ \ \forall\ \ \psi\ \in\ \Vo
\]
From the definition of KLS, we for any $\psi_0$, we have
\begin{align*}
\psi_0(x) = \Phi^{n_2}\psi_0(x) = \psi_0 (F^{n_2}(x))
\end{align*}
which implies $F^{n_2}$ is an identity map over $\fqn$. 
\begin{equation}
\label{eq:temp2}
\implies\ n_2 = rn_1\ \mbox{for some}\ r\ \in\ \mathbb{Z}_+
\end{equation}
From (\ref{eq:temp1}) and (\ref{eq:temp2}), $n_1 = n_2$
\end{proof}
As a finer observation from above lemma, the following relation between FSS and KLS in terms of divisors of periods of trajectories is obtained.
\newline
\begin{corollary}
The set of prime divisors of orbit lengths of a non-singular FSS is equal to the set of prime divisors of the orbit length of the associated KLS. 
\end{corollary}

The results discussed till now related the structure of solutions of the FSS (\ref{eq:FSS}) defined by the map $F$ with that of the dual linear system KLS (\ref{eq:Koopman}). These show that the orbit lengths (and chain lengths) of orbits of the FSS are a subset of the orbit lengths (and chain lengths) of the KLS. Importance of these results lies in locating possible compositional powers $k$ of the map iterates $F^{(k)}$ in solving for the points on orbits. An orbit of length $l$ exists for the FSS iff $F^{(l)}(x)=x$ has a solution. Since such problems are computationally hard, knowing for which possible $l$ a solution can be expected is of great importance. Above results facilitate knowing such possible $l$ by predicting them from the computation of orbit lengths of $\Phi$. 

However the above methodology of inferring the structure of solutions of FSS in terms of the linear system KLS is still unsatisfactory from computational point of view. This is because the KLS evolves over the state space $V^0$ which grows exponentially in size with respect to $n$, the dimension of the FSS. Hence an important problem is to find a $\Phi$-invariant subspace of $V^0$ which is smaller in dimension than $V^0$ and define a linear dynamical system on this subspace which will capture information on the structure of trajectories of FSS. Such a linear system is constructed next.

\section{Linear system which is reduced compared to KLS}
\label{sec:ROKLS}
In this section it is shown that a linear system of a minimal possible dimension can be constructed which has the complete information of the structure of all trajectories of the original FSS and is reduced in dimension relative to the KLS. This construction has the advantage that when the reduced system is of a much smaller size of $O(n^k)$ ($k\leq 3$) for an $n$-variable FSS (\ref{eq:FSS}) the computations are essentially equivalent to linear algebra computations. In such cases, the linear system required to represent information of the structure of trajectories of the FSS is no more of exponential size $O(q^n)$, which is the case in KLS.

Consider a given FSS (\ref{eq:FSS}) evolving over the state space $V$. Let the co-ordinate functions in the linear space $V^{0}$ of functions on $V$, be denoted by $\chi_{i}(.)$ i.e. for any $x$ in $V$ $\chi_{i}(x)=x_{i}$ is the $i$-th co-ordinate of $x$. Consider the $\Phi$-invariant subspace $W_{1}$ in $V^{0}$ to be the smallest cyclically generated subspace from the co-ordinate functions by action of $\Phi$. This space is the span of all iterates $\Phi^k(\chi_i)$ for $i=1,2,\ldots,n$ and $k=0,1,2,\ldots$.

Thus $W_{1}$ is the smallest $\Phi$-invariant subspace of $V^0$ which contains all co-ordinate functions $\chi_i$ for $i=1,\dots,n$. Let 
\[
N=\dim W_1
\]
Then there is a direct summand $W_{2}$ of $W_{1}$ in $V^0$ giving
\[
V^{0}=W_{1}\oplus W_{2}
\]
In fact it is well known that we can find a direct summand $W_{2}$ which is also $\Phi$-invariant but this is not necessary for our purpose here \cite{Hoffman,Strang}. Let 
\[
\Phi_{1}=\Phi|W_{1}
\]
be the restriction of $\Phi$ on $W_{1}$. Consider a basis of $W_{1}$ denoted by the ordered set
\[
\mathcal{B}=\{\psi_1,\psi_2,\ldots,\psi_N\}
\]
Since $\Phi_{1}(\psi_i)$ belongs to $W_1$ for every function $\psi_i$ in $\mathcal{B}$, for each $i$ there exist unique constants $k_{ji}$ in $\ff_q$ such that
\[
\Phi_{1}\psi_i=\sum_{j=1}^{N}k_{ji}\psi_j
\]
The matrix $K_1=[k_{ij}]$, $i,j=1,2,\ldots,N$ is thus the matrix representation of $\Phi_1$ in the basis $\mathcal{B}$. Similarly there is a unique linear representation of the co-ordinate functions $\chi_i$ (co-ordinate projections) in the basis $\mathcal{B}$. Denoting the $n$-tuple $[\chi_1,\ldots,\chi_n]^{T}$ by $\hat{\chi}$ and the $N$-tuple $[\psi_1,\ldots,\psi_N]^{T}$ by $\hat{\psi}$ there is a unique $n\times N$ matrix $C$ such that
\beq\label{statemap}
\hat{\chi}=C\hat{\psi}
\eeq
From these two matrices we can define a state, output dynamical system in $\fqN$ by
\begin{equation}
\label{ROKLS}
\hspace{0.3in}
\begin{aligned}
y(k+1)&=K_1y(k) \\
w(k)&=Cy(k)
\end{aligned}
\hspace{0.4in} 
\end{equation}
where $y(k)$ belongs $\fqN$ while $w(k)$ belongs to $\fqn$. We shall call the state space dynamics of the above system as \emph{Reduced Order Koopman Linear System} (ROKLS). We have the following structural relationship between the above system and the FSS (\ref{eq:FSS}).
\newline
\begin{theorem}
\label{thm:ROKLS}
Corresponding to a trajectory 
\begin{equation*}
x(0),x(1),\ldots,x(r-1)
\end{equation*}
of FSS (\ref{eq:FSS}) in $\fqn$, there is a unique trajectory $y(k)$ of the kind
\begin{equation}
\label{evaluation}
y(k)=\hat{\psi}(x(k)),\ \  k=0,1,\ldots,(r-1)
\end{equation}
such that for each $k$
\begin{equation*}
x(k)=Cy(k)    
\end{equation*}
If the trajectory $x(k)$ is periodic of period $r$ (chain of length $r$) then the above trajectory $y(k)$ is also periodic of period $r$ (respectively chain of length $r$).
\end{theorem}

Note that the uniqueness of the trajectory $y(k)$ in the statement of the above theorem refers to the kind of trajectory as $y(k)=\hat{\psi}(x(k))$ which is the trajectory obtained as evaluation of the $N$-tuple of functions $\hat{\psi}$ at each $x(k)$. 
\begin{proof}
Due to uniqueness of the $C$ matrix in expansion of the co-ordinate functions in (\ref{statemap}), a trajectory $x(k)$ in state space has unique expression as
\begin{equation}
\label{eq:x_C_chi}
x(k)=\hat{\chi}(x(k))=C\hat{\psi}(x(k))
\end{equation}
while the trajectory $y(k)$ of evaluations in (\ref{evaluation}) satisfies the following since $K_1$ is the matrix representation of $\Phi_1$ in the basis $\mcb$
\begin{align}
\begin{array}{lcl}
K_1y(k) & = & K_1\hat{\psi}(x(k))\\
 & = & \Phi_1 \psi_i(x(k))\ \ \mbox{ ($N$-tuple of $\Phi_1\psi_i$)}\\
 & = & \hat{\psi}(F(x(k))\\
 & = & \hat{\psi}(x(k+1))
\end{array}
\label{eq:thmROKLS}
\end{align}
Hence $y(k+1)=K_1y(k)$. This proves the first part of the theorem.

Now if $x(r)=x(0)$ then
\[
y(r)=\hat{\psi}(x(r))=\hat{\psi}(x(0))=y(0)
\]
Similarly if $x(r)$ is a fixed point $x(r+k)=x(r)$ for $k>0$ hence by above expression it follows that $y(r+k)=y(r)$. This proves the last two statements. 
\end{proof}
The RO-KLS as a dynamical system evolves over $\fqN$ and the correspondence of solution trajectories are over evaluation maps and not as functions. Given a FSS with an initial condition $x(0)$, the above theorem does not guarantee a function $g(x) \in W_1$ which has dynamics similar to that of $x(0)$ but a sequence of points $\hat{\psi}(x(k))$ with an initial condition $\hat{\psi}(x(0))$ obtained through the evaluation of basis vectors $\mcb$ whose dynamics under $K_1$ has a one-to-one correspondence with the dynamics of $x(0)$ under the FSS dynamics. The advantage being computation of dynamics of $\hat{\psi}(x(0))$ is well known from the linear theory. 
\newline
\begin{corollary}
The set of lengths of all trajectories of the FSS is a subset of the set of lengths of all trajectories of the RO-KLS
\end{corollary}

These lengths can be computed from the elementary divisors of $K_1$ using the theory developed in \cite{Gill1}.

\begin{theorem}
\label{cor:fixedpt}
A point $x_0$ is a fixed point of the FSS (\ref{eq:FSS}) iff 
there exists $y_0 \in \ff^N$, which is an eigenvector of $K_1$ corresponding to an eigenvalue $1$ satisfying 
\begin{enumerate}
    \item $y_0 = \hat{\psi}(C y_0) $ 
    \item $x_0 = C y_0$
\end{enumerate}
where $K_1$ and $C$ are defined as in equation (\ref{ROKLS}).
\end{theorem}
This theorem identifies all the fixed points of the FSS (\ref{eq:FSS}) through the eigenvectors for eigenvalue $1$ of the ROKLS matrix $K_1$. 
\begin{proof}
Given $x_0$ to be a fixed point of system (\ref{eq:FSS}) let $y_0 = \hat{\psi}(x_0)$. By theorem (\ref{thm:ROKLS}) it is proved that if $x_0$ is on an orbit of length $r$, then $y_0 = \hat{\psi}(x_0)$ also lies on an orbit of length $r$ under ROKLS. Since $x_0$ is a fixed point, it lies on an orbit of length 1 under FSS and hence  $y_0$ is also a fixed point of ROKLS which implies that under the ROKLS dynamics,
\[
K_1 y_0 = y_0
\]
and hence $y_0$ is a eigenvector of $K_1$ for an eigenvalue 1. Since by construction $y_0 = \hat{\psi}(x_0)$ and from equation (\ref{eq:x_C_chi}), it follows that 
\[
C y_0 = C \hat{\psi}(x_0) = x_0
\]
and 
\[
y_0 = \hat{\psi}(x_0) = \hat{\psi}(C y_0)
\]
which proves the necessity condition. 

To prove the converse, consider the set $\mathcal{Y} = \{ y_1, y_2,\dots, y_l\}$ of all eigenvectors of $K_1$ for the eigenvalue $1$. Given $y \in \mathcal{Y}$ which also satisfies $y = \hat{\psi}(Cy)$, compute $w = Cy$. So, 
\[
y = \hat{\psi}(Cy) = \hat{\psi}(w)
\]
We prove that this $w \in \ff^n$ is a fixed point of the FSS using the facts $y= K_1y$, $w = Cy$ and $y = \hat{\psi}(w)$
\begin{align*}
w = Cy &= C K_1 y \\
    &= C K_1 \hat{\psi}(w) \\
    &= C \hat{\psi}(F(w)) \hspace{0.7in} \mbox{due to equation (\ref{eq:thmROKLS}) as $K_1 \hat{\psi}(w) = \hat{\psi}(F(w))$}\\
    &= \hat{\chi}(F(w)) \hspace{1in} \mbox{due to equation (\ref{eq:x_C_chi})}  \\
    & = F(w)
\end{align*}
This proves that the $w$ constructed is a fixed point of FSS (\ref{eq:FSS}). 
\end{proof}

\begin{remark}
Computation of all the fixed points of a FSS is equivalent to computing all solutions of the polynomial equation 
\[
F(x) = x
\]
which is a $NP$-class computation problem. Theorem (\ref{cor:fixedpt}) converts this problem to an equivalent problem of computation of eigenvectors of a matrix representation of the restriction of the dual operator $\Phi$. An advantage of the computation of fixed points of the FSS by the conditions of the above theorem is that, among all the eigenvectors $v$ for the eigenvalue $1$, it is only required to verify which eigenvectors $v$ satisfy $v = \hat{\psi}(Cv)$ and not solve for it explicitly. 
\end{remark}

\begin{theorem}
A point $x_0$ is on an orbit of length $L$  under the FSS (\ref{eq:FSS}) iff 
there exists $y_0 \in \ff^N$ satisfying 
\begin{enumerate}
    \item $y_0$ is on an orbit of length $L$ under the dynamics of ROKLS. 
    \item $y_0 = \hat{\psi}(C y_0) $ 
    \item $x_0 = C y_0$
\end{enumerate}
\end{theorem}

\begin{proof}
Given a $x_0$ on an orbit of length $L$ under the FSS, let $y_0 = \hat{\psi}(x_0)$. 

By theorem (\ref{thm:ROKLS}), it is proved that $y_0$ is on an orbit of length $L$ under the ROKLS. Hence $K_1^L y_0 = y_0$.
Also, 
\[
x_0 = \hat{\chi}(x_0) = C \hat{\phi}(x_0) = Cy_0
\]
and
\[
y_0 = \hat{\psi}(x_0) = \hat{\psi}(Cy_0)
\]
which proves the necessary conditions. 

To prove sufficiency, let $y_0$ be on an orbit of length $L$ under the ROKLS satisfying $y_0 = \hat{\psi}(Cy_0)$. Let $x_0 = Cy_0$. Hence $y_0 = \hat{\psi}(x_0)$ and, 
\begin{align*}
F(x_0) &= \hat{\chi}(F(x_0)) \\    
    &= C \hat{\psi}(F(x_0)) \\
    &= C K_1 \hat{\psi}(x_0) \quad \quad \mbox{(from equation (\ref{eq:thmROKLS}))} \\
    &= C K_1 y_0
\end{align*}
Similarly, one can prove $F^{(m)} (x_0) = C K_1^m y_0$ for $m \geq 0$. Since $y_0$ is on an orbit of length $L$, $K_1^Ly_0 = y_0$ and hence  
\[
F^{(L)}(x_0) = C K_1^L y_0 = C y_0 = x_0
\]
Hence $x_0$ is on an orbit whose length divides $L$. To prove that the length is exactly $L$, assume the contrary. Let $l < L$ be the orbit length of $x_0$ under the FSS (i.e $F^{l}(x_0) = x_0$. From equation (\ref{eq:thmROKLS})  
\begin{align*}
    K_1^l y_0 &= K_1^l \hat{\psi}(x_0) \\
    &= \hat{\psi}(F^{(l)}(x_0)) \\
    &= \hat{\psi} (x_0) \\
    &= y_0
\end{align*}
which means that the orbit length of $y_0$ is also $l$ and that is a contradiction since $y_0$ is assumed to be on an orbit of length $L$. So, the orbit length of $x_0$ constructed as $x_0 = Cy_0$ is exactly $L$, the orbit length of $y_0$ whenever $y_0$ satisfies $y_0 = \hat{\psi}(Cy_0)$.
\end{proof}

\begin{theorem}
A point $x_0$ is a root of chain of length $L$ under FSS iff $\exists y_0 \in \ff^N$ satisfying 
\begin{enumerate}
    \item $y_0$ is the root of chain of length $L$ under ROKLS
    \item $y_0 = \hat{\psi}(Cy_0)$
    \item $x_0 = Cy_0$
\end{enumerate}
\end{theorem}

\begin{proof}
To prove necessity, let $x_0$ be a root of a chain of length $L$ under FSS. So $x_L = F^{L}(x_0)$ lies on a periodic orbit of length $M \geq 1$.

Let $y_0 = \hat{\psi}(x_0)$. From theorem (\ref{thm:ROKLS}), it known that if $x_0$ is on a chain of length $L$ under FSS, then $\hat{\psi}(x_0)$ is on a chain of length $L$ under ROKLS. Also, 
\[
x_0 = \hat{\chi}(x_0) = C \hat{\psi}(x_0) = Cy_0
\]
and 
\[
y_0 = \hat{\psi}(x_0) = \hat{\psi}(Cy_0)
\]

The only thing which needs to be proved that $y_0$ is the root of the chain. 

As defined in Theorem \ref{thm:ROKLS}, any trajectory $x(k)$ of FSS in $\ff$ is embedded in the state space $\ff^N$ of ROKLS (\ref{ROKLS}) by the map $x(k)\mapsto y(k)=\hat{\psi}(x(k))$. Since $x_0$ is a root of a chain of FSS iff there is no point $z$ in $\ff^n$ on a trajectory such that $F(z)=x_0$, by the above unique embedding of trajectories of FSS into trajectories of ROKLS, there is no point $\hat{\psi}(z)$ in $\ff^N$ on the trajectory of ROKLS such that $y_0=\hat{\psi}(x_0)=K(\hat{\psi}(z))$. Hence $y_0$ is also a root of the trajectory in the state space of ROKLS.

To prove the sufficiency, let $y_0$ be a root of a chain of length $L$ and $y_0 = \hat{\psi}(Cy_0)$ and $x_0 = Cy_0$. We need to prove $x_0$ is a root of chain of length $L$. By construction of $x_0$, 
$y_0 = \hat{\psi}(Cy_0) = \hat{\psi}(x_0)$. 

From the previous theorem, if $y_0 = \hat{\psi}(x_0)$, it is proved that
\[
F^{(m)}(x_0) = C K_1^m y_0
\]
Since $y_0$ is on a chain of length $L$, $K_1^Ly_0$ is on a period orbit of say length $M$. So $K_1^{L+kM}y_0 = K_1^{L}y_0$ for all $k \in \mathbb{Z}_+$. Hence 
\[
F^{(L+kM)}(x_0) = CK_1^{L+kM}y_0 = CK_1^{L}y_0 = F^{(L)} (x_0)
\]
which means $F^{(L)}(x_0)$ is on a periodic orbit. If there exists some $l < L$ such that $F^{(l)}(x_0)$ is on a periodic orbit then 
\[
F^{(l+m_1)} (x_0) = F^{(l)}(x_0)
\]
for some $m_1 \in \mathbb{Z}_+$. Since $y_0 = \hat{\psi}(x_0)$ 
\begin{align*}
    K_1^{(l+m_1)}y_0 &= K_1^{(l+m_1)}\hat{\psi}(x_0) \\
    &= \hat{\psi}( F^{(l+m_1)}(x_0)) \\
    &= \hat{\psi}( F^{(l)} (x_0)) \\
    &= K_1^l y_0
\end{align*}
which proves that $K_1^l y_0$ lies on a periodic orbit and the length of chain starting from $y_0$ under ROKLS is $l$ which is a contradiction. So the length of the chain starting from $x_0$ under FSS is also $L$. 

The last thing to prove is that $x_0$ is the root of the chain. Assume the contrary again. Let there be $x \in \ff^n$ such that $F(x) = x_0$. Construct $y = \hat{\psi}(x)$. 
\begin{align*}
    K_1 y &= K_1 \hat{\psi}(x) \\
        &= \hat{\psi}(F(x)) \\
        &= \hat{\psi}(x_0) \\
        &=  y_0
\end{align*}
which proves that $y_0$ is also not a root which is a contradiction. Hence $x_0$ is the root of the chain of length $L$ under FSS.

\end{proof}

\subsection{Algorithm for computation of ROKLS}

Given an initial condition $x(0)$ of the FSS, the RO-KLS gives an explicit computational approach to compute the length of the orbit (or chain) starting from $x(0)$. This is achieved by setting $y(0) = \hat{\psi}(x(0))$ and then computing the solution of $y(0)$ under the dynamics of RO-KLS and use the equivalence proved in theorem \ref{thm:ROKLS}. Construction of the cyclic $\Phi$-invariant subspace $W_1$ is shown in Algorithm \ref{alg:KLSCyc}. Once the cyclic invariant subspace $W_1$ and its basis $\mcb$ is computed, the matrix representation of $K_1$ follows easily.

\begin{algorithm}[h]
\begin{algorithmic}[1]
\caption{Construction of $W_1$ - Cyclic invariant subspaces spanning $\chi_i(x)$}
\label{alg:KLSCyc}
\Procedure{Cyclic Invariant subspace}{$W_1$}
\State \textbf{Outputs}: 
\begin{itemize}
    \item[] $W_1$ - the cyclic subspace which span the coordinate functions and $\Phi$-invariant 
    \item[] $\mathcal{B}$ - the basis for the cyclic subspace $W_1$
\end{itemize}
\State Compute the cyclic Subspace  
\Statex $Z(\chi_1; \Phi) = \langle \chi_1,  \Phi \chi_1,\dots, \Phi^{l_1-1} \chi_1 \rangle$
\State Set of basis functions $\mathcal{B} = \{\chi_1,\Phi \chi_1,\dots,\Phi^{l_1-1}\chi_1 \}$
\If{$\chi_2,\chi_3,\dots,\chi_n \in \mbox{Span}\{\mathcal{B}\}$} 
\State $W_1 \gets \mbox{Span}\{\mathcal{B}\}$
\State \textbf{halt}
\Else
\State Find the smallest $i$ such that $\chi_i \notin \mbox{span}\{\mathcal{B}\}$
\State Compute the smallest $l_i$ such that 
\Statex $\Phi^{l_i} \chi_i \in \mbox{Span}\{ \mathcal{B} \cup \langle \chi_i,\Phi \chi_i,\dots, \Phi^{l_i-1}\chi_i \rangle\} $
\State $V_i = \{\chi_i,\Phi \chi_i,\dots, \Phi^{l_i-1} \chi_i \}$
\State Append the set $V_i$ to $\mathcal{B}$
\State \textbf{go to} 5
\EndIf
\EndProcedure
\end{algorithmic}
\end{algorithm}

\subsection{Numerical Example}
Consider a biochemical network represented by the following Boolean equations (\cite{Goodwin})
\begin{equation}
\begin{aligned}
A(k+3) &= A(k)B(k+1) + 1 \\
B(k+3) &= A(k+1)B(k) + 1
\end{aligned}
\end{equation}
This dynamics can be represented in terms of a 6-state dynamical system over $\ff_2$ as follows
\begin{equation}
\begin{bmatrix} x_1(k+1) \\ x_2(k+1) \\ x_3(k+1) \\ x_4(k+1) \\ x_5(k+1) \\ x_6(k+1) \end{bmatrix} = \begin{bmatrix} x_2(k) \\ x_3(k) \\ x_1(k)x_5(k) + 1 \\ x_5(k) \\ x_6(k) \\ x_2(k)x_4(k)+1  \end{bmatrix} 
\end{equation}
where $x_1(k) = A(k)$ and $x_4(k) = B(k)$ respectively and the state map $F(x)$ is defined on the right. RO-KLS is constructed and the solutions of the FSS are analyzed through the solutions of the RO-KLS. It can be seen that the dimension of the cyclic invariant subspace is much smaller than $2^6$. 

Constructing the Koopman subspace as in Algorithm \ref{alg:KLSCyc}, the cyclic evolution of the co-ordinate function $\chi_1$ is
\begin{eqnarray}
\label{eq:Cycsub1}
\begin{aligned}
&\chi_1 \to \chi_2 \to \chi_3 \to \chi_1\chi_5+1 \to \chi_2\chi_6+1 \to \chi_3(\chi_2\chi_4+1) + 1 \to \chi_1 \chi_5(\chi_3+1) + \chi_3\chi_5 \\
&\to \chi_1 \chi_5 \chi_6 (\chi_2+1) + \chi_6 \to \chi_3\chi_5 + \chi_2\chi_6 + \chi_2 \chi_6(\chi_3 + \chi_3\chi_4+\chi_4) + 1 \to \chi_3 \chi_5 + 1 \\
&\to \chi_6(\chi_1\chi_5+1) + 1 \to \chi_2 \chi_4(\chi_6+1) + \chi_2 \chi_6 \to \chi_3+\chi_2 \chi_3 \chi_4(\chi_5+1) \\
&\to 1 + \chi_1\chi_5 + \chi_3\chi_5(1+\chi_1+\chi_6+\chi_1\chi_6) \to \chi_2 \chi_6 + 1
\end{aligned}
\end{eqnarray}
where the $\to$ represents the operation
\[
\psi(x) \to \psi \circ F(x)
\]
The last function in the sequence $\chi_2\chi_6+1$ is a linear combination of the previous functions (as it had already appeared in the sequence before) and we have a cyclic subspace of $\chi_1$. Also we see that $\chi_2$ and $\chi_3$ are already in this subspace. We construct the cyclic subspace of $\chi_4$
\begin{eqnarray}
\label{eq:Cycsub2}
\chi_4 \to \chi_5 \to \chi_6 \to \chi_2\chi_4 + 1 \to \chi_3\chi_5+1
\end{eqnarray} 
where $\chi_3\chi_5+1$ is already in the cyclic subspace of $\chi_1$. These two sequences span all the basis functions and $W_1$ is the span of functions in  (\ref{eq:Cycsub1}) and (\ref{eq:Cycsub2}). The matrix representation of $K_1$ is omitted due to space constraints. Note that RO-KLS is a linear system of dimension 18, while the full order KLS is of dimension 64 (since $\mbox{dim}(\Vo) = 2^6$). To analyze the orbits of the RO-KLS, the minimal polynomial of RO-KLS is computed
\[
p(\xi) = \xi^4\ (\xi+1)^2\ (\xi^4+\xi^3+\xi^2+\xi+1)^2
\]
The RO-KLS can be decomposed into non-singular and nilpotent part where the non-singular part correspond to the periodic orbits and the nilpotent part correspond to the chains. The non-singular part has a minimal polynomial $(\xi+1)^2(\xi^4+\xi^3+\xi^2+\xi+1)^2$ which corresponds to possible orbit lengths of $1,2,5\ \mbox{and}\ 10$ while the degree of nilpotence is $4$ which corresponds to the length of the longest chain. (Details regarding computation of solutions of linear FSS is developed in (\cite{Gill1})). 
\begin{enumerate}
    \item The original system has one orbit each of length 2,5 and 10. The RO-KLS also has predicted orbit lengths of 2,5 and 10.
    \item The longest chain in the original system is of length 4 which is also in accordance with the results from RO-KLS. 
\end{enumerate}
\section{Observability and Observer theory for FSS using RO-KLS}
Consider an FSS of equation (\ref{eq:FSS}) which is reproduced here for convenience. 
\begin{equation*}
\begin{aligned}
    x(k+1) &= F(x(k)) \\
    z(k) &= g(x(k))
\end{aligned}
\end{equation*}
where $F : \fqn \to \fqn $ is the state transition map and $g : \fqn \to \ff_q^m$ is the output map. Similar to the Koopman Linear System developed for the system developed in the previous section, one can associate a Koopman linear system for the system for (\ref{eq:FSS}), where $\Phi$ is the Koopman operator corresponding to the FSS. Consider a sequence of outputs $z(0),z(1),\dots,z(L)$ of (\ref{eq:FSS}) corresponding to an initial condition $x(0)$ in $\fqn$. We recall as defined in introduction, this system (\ref{eq:FSS}) is said to be \textit{observable} if given a sequence of its outputs $z(0),z(1),\dots,z(L)$ for some $L$, there exists a unique initial condition $x(0)$ which generates the sequence of output. We shall refer to this problem of computing the initial condition given an output sequence as the observability problem.

The problem of reconstruction of initial condition from the sequence of outputs for a general non-linear system involves solving the polynomial system of equations
\begin{equation}
\label{eq:opseq}
z(k)=g(F^k(x(0))) = \Phi^k (g)(x(0))
\end{equation}
for $x(0)$. This is a well known hard computational problem for nonlinear FSS. In the case of FSS over $\ff_2$, this is the problem of solving all satisfying assignments of the Boolean system for $x(0)$. This problem is known to be of class $NP$. Hence unique solvability of (\ref{eq:opseq}) is the necessary and sufficient condition for observability of the system (\ref{eq:FSS}).

We construct an observability condition for (\ref{eq:FSS}) in terms of the matrix $K_1$ of the RO-KLS and a matrix representation of $g$ in terms of the basis of $W_1$. 


\subsection{RO-KLS for FSS with outputs}
The concept of RO-KLS is extended for systems with outputs (\ref{eq:FSS}) in following way. Recall the cyclic invariant subspace $W_1$ which is spanned by the co-ordinate functions $\chi_i$ is computed as in algorithm \ref{alg:KLSCyc}. Let the output map be defined as 
\[
g(x) = \begin{bmatrix} g_1(x) \\ g_2(x) \\ \vdots \\ g_m(x)\end{bmatrix}
\]
where each $g_i(x)$ is a $\ff_q$-valued function. The space $W_1$ is now expanded as sum of cyclic invariant subspaces of $\Phi$ generated by $\chi_i$ as well as $g_i$. Let this resulting space be denoted $W(g)$ which is the smallest $\Phi$-invariant subspace of $V^0$ which contains $\{\chi_i\}\cup \{g_j\}$ for $i=1,\dots,n$ and $j=1,\dots,m$. Let $\mcb$ be a basis for $W(g)$.
\[
\mcb = \{\psi_1(x),\psi_2(x),\dots,\psi_N(x) \}
\]
Since each of $g_i(x) \in W(g)$, there exists a unique representation of $g_i(x)$ in terms of the basis $\mcb$.
\[
g_i(x) = \sum_{j=1}^N \gamma_{ij} \psi_j(x)
\]
The output map $g(x)$ then can be represented as 
\begin{equation}
\begin{aligned}
g(x) = \Gamma \hat{\psi}(x)
\end{aligned}
\label{outputmap}
\end{equation}
where the entries of $\Gamma$ are defined as $\Gamma(i,j) = \gamma_{ij}$ and $\hat{\psi} = [\psi_1(x),\psi_2(x),\dots,\psi_N(x)]^T$. Consider the following dynamical system 
\begin{equation}
\label{eq:ROKLS-op}
    \begin{aligned}
    y(k+1) &= K_1 y(k) \\
    w(k) &= C y(k) \\
    \yop(k) &= \Gamma y(k)
    \end{aligned}
\end{equation}
where $y(k) \in \fqN$, $w(k) \in \fqn$, $\yop(k) \in \ff_q^m$, $K_1$ is the restriction of $\Phi$ on the space $W(g)$ and $C$ is the map defined as in (\ref{statemap}). This system is defined as the RO-KLS for an FSS with output with states $y(k)$, coordinate evaluation $w(k)$ and output evaluation $\yop(k)$. $K_1$ and $C$ can be viewed as the state transition map, coordinate evaluation map as in (\ref{ROKLS}) and $\Gamma$ is the output map.
\newline
\begin{lemma}
\label{lem:OutputEquivalence}
Consider an FSS as in equation (\ref{eq:FSS}) with initial condition $x(0)$ and a RO-KLS defined as in equation (\ref{eq:ROKLS-op}) with initial condition $y(0) = \hat{\psi}(x(0))$. If the FSS has an output sequence $z(0),z(1),\dots$ then the corresponding output sequence of RO-KLS is $\yop(0) = z(0), \yop(1) = z(1), \dots$
\end{lemma}

The lemma assures that the output sequences of both the FSS and RO-KLS are the same when the initial condition of the RO-KLS is $y(0) = \hat{\psi}(x(0))$. 

\begin{proof}
  From theorem \ref{thm:ROKLS} it is proved there exists a unique trajectory $y(k) = \phat(x(k))$ of the RO-KLS for each trajectory $x(k)$ of the KLS. The output of the RO-KLS is 
\[
\yop(k) = \Gamma y(k) = \Gamma \phat(x(k))
\]
From equation (\ref{outputmap}), 
\[
\Gamma \phat(x(k)) = g(x(k)) = z(k)
\]
Combining the above equations $\yop(x) = z(k)$
\end{proof}

\begin{theorem}
Consider a FSS with output as in equation (\ref{eq:FSS}) and the RO-KLS as in equation (\ref{eq:ROKLS-op}). Then the system (\ref{eq:FSS}) is observable if the linear system $(K_1,\Gamma)$ is observable. 
\end{theorem}
\begin{proof}
  Given a sequence of outputs $z(k)$, let $\yop(k) = z(k)$. Computation of initial condition $x(0)$ of the FSS (\ref{eq:FSS}) is equivalent to computing $y(0)$ of the RO-KLS (\ref{eq:ROKLS-op}) by assigning $\yop(k) = z(k)$ and computing $y(0)$ of the RO-KLS. If $y(0)$ is uniquely determined, $x(0)$ is uniquely determined by (\ref{statemap}) by assigning $\hat{\psi} = y(0)$. Writing the output of the RO-KLS at each instant
  \begin{equation*}
  \begin{aligned}
      \yop(0) &= \Gamma y(0) \\
      \yop(1) &= \Gamma y(1) = \Gamma K_1 y(0) \\
      \yop(2) &= \Gamma y(2) = \Gamma K_1^2 y(0) \\
      &\ \ \vdots \\
      \yop(N-1) &= \Gamma y(N-1) = \Gamma K_1^{N-1} y(0)
 \end{aligned}
  \end{equation*}
  which can be written as 
  \begin{equation}
  \label{eq:ObsIC}
     \begin{bmatrix} \yop(0) \\ \yop(1) \\ \vdots \\ \yop(N-1) \end{bmatrix} = \begin{bmatrix} \Gamma \\ \Gamma K_1 \\ \vdots \\ \Gamma K_1^{N-1} \end{bmatrix} y(0) =: \mathcal{O} y(0)  
   \end{equation}
An unique solution for $y(0)$ exists if $\mathcal{O}$ is of full rank which in linear systems theory parlance is equivalent to saying that the pair $(K_1,\Gamma)$ is observable.
\end{proof}

\begin{remark}
Perhaps the most important consequence of this theorem is that the observability of the non-linear system with output (\ref{eq:FSS}) is translated to a condition of observability of the RO-KLS. Whenever the linear system (\ref{eq:ROKLS-op}) is observable and the dimension of $W(g)$ is small enough, the nonlinear observability can be computed by a feasible linear algebra computation.  
\end{remark}

It is thus logical to explore the next step, whether and how we can build an observer for (\ref{eq:FSS}) in terms of the linear system (\ref{eq:ROKLS-op}).

\subsection{State Observer for Non-Linear Finite State Systems}
 In linear system theory the Luenberger observer can compute the internal state of the system $x(k)$ from output measurements $z(k)$. We now show that such an observer is built for the linear system (\ref{eq:ROKLS-op}) when it is observable can compute the initial condition $x(0)$ of the non-linear system (\ref{eq:FSS}). Let $\hy(k)$ be the states of the observer. The observer dynamics is given by
\begin{equation}
\label{eq:ROKLSobs}
    \begin{aligned}
    \hy(k+1) &= K_1 \hy(k) + L (\yop(k) - \hy_{op}(k)) \\
    \hy_{op}(k) &= \Gamma \hy(k)
    \end{aligned}
\end{equation}
Let $e(k) = y(k)-\hy(k)$. The dynamics of $e(k)$ is given as 
\begin{equation*}
\begin{aligned}
    e(k+1) &= y(k+1) - \hy(k) \\
            &= K_1 y(k) - \bigg(K_1 \hy(k) + L(\yop(k) - \hy_{op}(k))\bigg) \\
            &= K_1 y(k) - \bigg(K_1 \hy(k) + L (\Gamma y(k) - \Gamma \hy(k)) \bigg) \\
            &= (K_1 - L \Gamma) (y(k) - \hy(k)) \\ 
            &= (K_1 - L \Gamma) e(k)
\end{aligned}
\end{equation*}
Note that $e(k)$ is the error between the state value and the observer state. Once the internal state $\hy(k)$ of the observer converges to the state $y(k)$ of the RO-KLS, the internal states of the FSS is computed through (\ref{statemap}) as 
\[
x(k) = C \hy(k)
\]
The idea is to make this error $e(k)$ go to a $zero$ matrix by making $(K_1 - L \Gamma)$ to be nilpotent by choosing an appropriate $L$. Once such a $L$ is chosen, the error dynamics reaches $zero$ in a maximum of $r$ time instants, where $r$ is the index of nilpotence of the matrix $(K_1 - L \Gamma)$. Such an $L$ can always be chosen if the pair $(K_1,\Gamma)$ is observable. 

Unlike the observer theory for systems over reals where the observer state converges to the internal state asymptotically, the state $\hy(k)$ converges to $y(k)$ in finite time.

\subsubsection{Observable and Unobservable modes}
Consider a LFSS as below 
\begin{equation*}
    \begin{aligned}
    x(k+1) &= A x(k) \\
    y(k) &= C x(k)
    \end{aligned}
\end{equation*}
where $x(k)$ and $y(k)$ are states and outputs respectively.
It is well known from the linear theory of dynamical systems that there exists a similarity transformation $P$ \cite{Kailath,Wonham} which decomposes the system matrix into a canonical from as 
\begin{equation}
\begin{aligned}
\overline{A} &= \begin{bmatrix} \overline{A}_{11} & 0 \\ \overline{A}_{21} & \overline{A}_{22} \end{bmatrix} \\
\overline{C} &= \begin{bmatrix} \overline{C}_{1} & 0 \end{bmatrix} 
\end{aligned}
\label{eq:transLFSSop}
\end{equation}
where the pair $(\overline{A}_{11},\overline{C}_1)$ is observable. The eigenvalues of $\overline{A}_{11}$ and $\overline{A}_{22}$ are called as observable and unobservable modes of the system. In the theory for linear dynamical systems over finite fields, the system dynamics is analyzed in terms of the elementary divisors of the state transition matrix. So, the notion of observable and unobservable modes are redefined in terms of elementary divisors of the matrix. 
\newline
\begin{definition}
Given a transformed linear finite state system with system matrices as in (\ref{eq:transLFSSop}), the elementary divisors of $\overline{A}_{11}$ are the \emph{observable elementary divisors} and the elementary divisors of $\overline{A}_{22}$ are the \emph{unobservable elementary divisors}. 
\end{definition}
The RO-KLS can be decomposed in a similar way into
\begin{align}
    \overline{K}_1 = \begin{bmatrix} \overline{K}_{11} & 0 \\ \overline{K}_{21} & \overline{K}_{22} \end{bmatrix} \ \ \ \overline{\Gamma} = \begin{bmatrix} \overline{\Gamma}_{1} & 0 \end{bmatrix}
    \label{eq:TransROKLS}
\end{align}
where $\overline{K}_{11}$ is the observable part and $\overline{A}_{22}$ is the unobservable part corresponding to the observable and unobservable elementary divisors respectively.
\newline
\begin{definition}
An LFSS with transformed system matrices $\bar{A}$ and $\bar{C}$ as in equation (\ref{eq:transLFSSop}) is said to be \emph{detectable} if $\bar{A}_{22}$ is nilpotent
\end{definition}
The following theorem characterizes the condition under which a dynamic observer can be built for a FSS using the RO-KLS framework. 
\newline
\begin{theorem}
Given a FSS as in equation (\ref{eq:FSS}) and its corresponding RO-KLS as in equation (\ref{eq:ROKLS-op}), there exists a $L$ such that the states of dynamic observer defined in (\ref{eq:ROKLSobs}) converges to $\hy(k) = \phat(x(k))$ if the pair $(K_1,\Gamma)$ is detectable.
\end{theorem}
\begin{proof}
  Assuming the RO-KLS to be in the canonical form as in equation (\ref{eq:TransROKLS}) and $\overline{L} = \begin{bmatrix} \overline{L}_1 \\ \overline{L}_2 \end{bmatrix}$, the error $e(k)$ has the following dynamics 
  \begin{equation}
      e(k+1) = \begin{bmatrix} \overline{K}_{11} - \overline{L}_1 \overline{\Gamma}_1 & 0 \\ \overline{K}_{21} - \overline{L}_2 \overline{\Gamma}_1 & \overline{K}_{22}  \end{bmatrix} e(k)
  \end{equation}
  The matrix $(\overline{K}_{11} - \overline{L}_1 \overline{\Gamma}_1)$ can be made a nilpotent matrix by choosing $\overline{L}_1$. Such a $\overline{L}_1$ always exists because the pair $(\overline{K}_{11},\overline{\Gamma}_1)$ is observable. So the error dynamics settles down to $zero$ if $\overline{K}_{22}$ is a nilpotent matrix.
\end{proof}
\subsection{Observer Construction for FSS}
This section gives an algorithm to construct an observer for the FSS. Whenever $K_{22}$ as in decomposition (\ref{eq:TransROKLS}) correspond to a nilpotent matrix, the states of the observer $\hy(k)$ settles down to $y(\phat(x(k))$
\begin{algorithm}
\begin{algorithmic}[1]
\caption{Construction of Dynamic Observer through RO-KLS}
\label{alg:ROKLSobs}
\Procedure{Dynamic Observer}{}
\State \textbf{Outputs}: 
\begin{itemize}
    \item[] $\hy(k)$ - the computed internal state of the RO-KLS 
    \item[] $\hat{x}(k)$ - the computed internal state of the non-linear FSS
\end{itemize}
\State Compute $W_1$, the cyclic invariant subspace spanning the basis functions $\chi_i(x)$ and the output functions $g_i(x)$.
\State Compute the matrices $K_1$, $C$ and $\Gamma$ using a basis $\mcb$ of the space $W_1$
\State Construct the RO-KLS as in equation (\ref{eq:ROKLS-op}).
\State Compute the transformation $\bar{y}(k) = P^{-1} \hy(k)$ which decomposes $K_1$ and $\Gamma$ as in equation (\ref{eq:TransROKLS}). 
\State Find $\overline{L} = \begin{bmatrix} \overline{L}_1 \\ 0 \end{bmatrix}$ such the matrix $(\overline{K}_{11} - \overline{L}_1\overline{\Gamma}_1)$ is a nilpotent matrix. 
\State Construct the dynamic observer as in equation (\ref{eq:ROKLSobs}) with $L = P^{-1} \overline{L}$
\State The state $\hy(k)$ is the internal state of the RO-KLS
\State Compute $\hat{x}(k) = C \hy(k)$. 
\EndProcedure
\end{algorithmic}
\end{algorithm}

\subsection{Numerical Example}
In this section, a numerical example is presented where the initial condition of the FSS is estimated from the sequence of outputs using the observability matrix of RO-KLS. Also, a dynamic observer is constructed where the internal state of the observer converges to the internal state of RO-KLS.

Consider the FSS over $\ff_3^2$ as follows
\begin{align}
\begin{aligned}
    \begin{bmatrix} x_1(k+1) \\ x_2(k+1)\end{bmatrix} &= \begin{bmatrix} 2 & 1 \\ 1 & 1 \end{bmatrix} \begin{bmatrix} x_1(k) \\ x_2(k)\end{bmatrix} \\
    z(k) &= x_1(k)^2+x_2(k)
    \end{aligned}
    \label{eq:Obsex}
\end{align}
The state transition map is linear, but the output map ($\ff_3^2 \to \ff_3$) is non-linear. $W_1$ was computed to be of dimension $4$. Considering the following basis $\mcb$ for $W_1$
\[
\mcb = \{ \chi_1, \chi_2, \chi_1^2, \chi_1^2+\chi_2^2+\chi_1\chi_2 \}
\]
the matrix $K_1$ and $\Gamma$ are given as 
\[
K_1 = \begin{bmatrix} 2 & 1 & 0 & 0 \\ 1 & 1 & 0 & 0 \\ 0 & 0 & 0 & 1 \\ 0 & 0 & 1 & 0\end{bmatrix}  \hspace{0.3in} \Gamma = \begin{bmatrix} 0 & 1 & 1 & 0\end{bmatrix}
\]
the co-ordinate map is given as 
\[
C = \begin{bmatrix} 1 & 0 & 0 & 0 \\ 0& 1 & 0 & 0\end{bmatrix}
\]
and the observability matrix is given as 
\[
\mathcal{O} = \begin{bmatrix}
0 & 1 & 1 & 0 \\
 1 & 1 & 0 & 1 \\
 0 & 2 & 1 & 0 \\
 2 & 2 & 0 & 1
\end{bmatrix}
\]
This system is observable as the observability matrix $\mathcal{O}$ is full rank. 
\subsubsection{Computation of initial condition}
Consider the following output sequence for $z(k)$ 
\[
z(0) = 1, \hspace{0.1in} z(1) = 0, \hspace{0.1in} z(2) = 1, \hspace{0.1in} z(3) = 2
\]
The linear system of equations are written as in equation (\ref{eq:ObsIC})
\[
\begin{bmatrix} 1 \\ 0 \\ 1 \\ 2 \end{bmatrix} =  \begin{bmatrix}
0 & 1 & 1 & 0 \\
 1 & 1 & 0 & 1 \\
 0 & 2 & 1 & 0 \\
 2 & 2 & 0 & 1
\end{bmatrix} y(0)
\]
which gives the initial condition of RO-KLS to be 
\[
y(0) = \begin{bmatrix} 2 & 0 & 1 & 1\end{bmatrix}^T
\]
and the initial condition $x(0)$ of the FSS is given by 
\[
x(0) = Cy(0) = \begin{bmatrix} 2 \\ 0\end{bmatrix}
\]
\subsubsection{Dynamic Observer}
The following $L$ matrix makes $K_1 - L\Gamma $ nilpotent.
\[
L = \begin{bmatrix} 1 & 0 & 0 & 2\end{bmatrix}^T
\]
The observer dynamics is given by 
\begin{align*}
\hy(k+1) &= (K_1 - L \Gamma) \hy(k) + L z(k) \\
&= \begin{bmatrix} 2 & 0 & 2 & 0 \\
1 & 1 & 0 & 0 \\
0 & 0 & 0 & 1 \\
0 & 1 & 2 & 0 \end{bmatrix} \hy(k) + \begin{bmatrix} 1 \\ 0 \\ 0 \\ 2\end{bmatrix} z(k)
\end{align*}
The minimal polynomial of $K_1 - L \Gamma$ is $x^4$ and so the observer state $\hy(k)$ converge to the $y(k)$, the state of RO-KLS from any arbitrary initial condition in a maximum of $4$ time instances. After the internal state of observer converges to that of RO-KLS, the internal state of the FSS can be computed by using co-ordinate map $C$ as 
\[
x_{obs}(k) = C\hy(k)
\]
Starting with an initial condition $[2,0]^T$ for the FSS, the following output sequence is obtained 
\[
1,0,1,2,1,0,1,2,1,0,\dots
\]
Initializing the observer to $[0,0,0,0]^T$, the following table compares the internal state of the FSS $x(k)$ with the predicted state by observer $x_{obs}(k)$.
\begin{center}
\begin{tabular}{|c|c|c|c|}
\hline
 $k$ & $z(k)$ & $x(k)$ & $x_{obs}(k)$ \\ [0.5ex] 
 \hline
 0 & 1 & $[2,0]^T$ & $[0,0]^T$ \\
 \hline
 1 & 0 & $[1,2]^T$ & $[1,0]^T$ \\
 \hline
 2 & 1 & $[1,0]^T$ & $[2,1]^T$ \\
 \hline
 3 & 2 & $[2,1]^T$ & $[0,0]^T$ \\
 \hline
 4 & 1 & $[2,0]^T$ & $[2,0]^T$ \\
\hline
 5 & 0 & $[1,2]^T$ & $[1,2]^T$ \\
 \hline
 6 & 1 & $[1,0]^T$ & $[1,0]^T$ \\
 \hline
 7 & 2 & $[2,1]^T$ & $[2,1]^T$ \\
 \hline
 8 & 1 & $[2,0]^T$ & $[2,0]^T$ \\
 \hline
 9 & 0 & $[1,2]^T$ & $[1,2]^T$ \\
 \hline
\end{tabular}
\end{center}
As expected, $x_{obs}(k)$ converges to the internal state $x(k)$ of the FSS at $k = 4$.
\section{Conclusion}
The Koopman linear system (KLS) is the linear system defined by the dual map of the state update of an FSS on the space of functions. Although of an exponential size, KLS can be used to infer the structure of solutions of the FSS. A reduced order KLS (RO-KLS) of a possibly smaller dimension is constructed by generating the cyclic invariant subspace containing the coordinate functions. When the RO-KLS is of significantly small dimension, the computation of the structure of solutions of non-linear FSS can be significantly simplified and solved by using tools from linear algebra. Without the use of KLS or the reduced system, these problems belong to hard problem classes of computation. As an extension of the theory, RO-KLS is constructed for systems with outputs. It it shown that the non-linear FSS is observable iff the RO-KLS is an observable linear system. A Luenberger type observer is also constructed using the RO-KLS and it is shown how it can recover the internal state of the non-linear FSS. Without an observer, computation of the internal state from outputs of a non-linear FSS involves hard computational problems. Such an approach to compute the internal states as well as the initial condition should be immensely useful in the field of cryptography and systems biology. Such applications can be pursued as future directions of this work.


%
%



\bibliographystyle{abbrv}
\bibliography{main.bib}

\end{document}